\documentclass[sigconf]{aamas} 

\usepackage{amsmath,amsthm}
\usepackage{algorithm}
\usepackage[noend]{algpseudocode}
\usepackage{algorithmicx}
\usepackage{array}
\newcolumntype{P}[1]{>{\centering\arraybackslash\hspace{0pt}}p{#1}}
\usepackage{booktabs}
\usepackage{cleveref}
\usepackage{comment}
\usepackage[inline]{enumitem}
\usepackage{fancyhdr}
\usepackage{float}
\usepackage{graphicx}
\usepackage{latexsym}
\usepackage{lineno}
\usepackage{multicol}
\usepackage{multirow}
\usepackage{rotating}
\usepackage{thmtools}
\usepackage{thm-restate}
\usepackage{url}

\algrenewcommand\algorithmicindent{0.75em}
\makeatletter
\newcommand{\algmargin}{\the\ALG@thistlm}
\makeatother
\newlength{\whilewidth}
\algdef{SE}[parWHILE]{parWhile}{EndparWhile}[1]
{\parbox[t]{\dimexpr\linewidth-\algmargin}{%
		\hangindent\whilewidth\strut\algorithmicwhile\ #1\ \algorithmicdo\strut}}{\algorithmicend\ \algorithmicwhile}%
\algnewcommand{\parState}[1]{\State%
	\parbox[t]{\dimexpr\linewidth-\algmargin}{\strut #1\strut}}

\declaretheorem[name=Theorem]{thm}
\newtheorem{ex}{Example}

\newcommand{\am}{\text{AM}}

\newcommand{\cvm}{\odot}
\newcommand{\defended}{{\sc Defended~}}

\newcommand{\existence}{{\sc Existence~}}
\newcommand{\m}{\text{M}}
\newcommand{\majority}{{\sc Majority~}}

\newcommand{\mm}{\mathcal M}

\newcommand{\Omit}[1]{}
\newcommand{\perdefended}{{\sc \%Defended~}}
\newcommand{\perundefended}{{\sc \%Undefended~}}
\newcommand{\Rho}{\mathrm{P}}

\newcommand{\undefended}{{\sc Undefended~}}
\newcommand{\uniform}{{\sc Uniform~}}
\newcommand{\vc}{\vec{c}}
\newcommand{\vd}{\vec{d}}
\newcommand{\ve}{\vec{e}}

\newcommand{\veta}{\vec{e}}
\newcommand{\vetap}{\vec{d}}

\newcommand{\vm}{\vec{m}}
\newcommand{\vpi}{\vec\pi}
\newcommand{\vr}{\vec{r}}
\newcommand{\vrho}{\vec\rho}

\newcommand{\vone}{\vec{1}}

\usepackage{balance} 

\setcopyright{ifaamas}
\acmConference[AAMAS '22]{Proc.\@ of the 21st International Conference
on Autonomous Agents and Multiagent Systems (AAMAS 2022)}{May 9--13, 2022}
{Online}{P.~Faliszewski, V.~Mascardi, C.~Pelachaud,
M.E.~Taylor (eds.)}
\copyrightyear{2022}
\acmYear{2022}
\acmDOI{}
\acmPrice{}
\acmISBN{}

\acmSubmissionID{fp752}

\title[Anti-Malware Sandbox Games]{Anti-Malware Sandbox Games}

\author{Sujoy Sikdar}
\affiliation{
  \institution{Binghamton University}
  \city{Binghamton}
  \state{NY}
  \country{USA}}
\email{ssikdar@binghamton.edu}

\author{Sikai Ruan}
\affiliation{
  \institution{Rensselaer Polytechnic Institute}
  \city{Troy}
  \state{NY}
  \country{USA}}
\email{ruans2@rpi.edu}

\author{Qishen Han}
\affiliation{
  \institution{Rensselaer Polytechnic Institute}
  \city{Troy}
  \state{NY}
  \country{USA}}
\email{hanq6@rpi.edu}

\author{Paween Pitimanaaree}
\affiliation{
  \institution{SCB Securities Co. Ltd.}
  \city{London}
  \country{UK}}
\email{paween.pit@gmail.com}

\author{Jeremy Blackthorne}
\affiliation{
  \institution{Boston Cybernetics Institute}
  \city{Cambridge}
  \state{MA}
  \country{USA}}
\email{jeremy.blackthorne@gmail.com}

\author{Bulent Yener}
\affiliation{
  \institution{Rensselaer Polytechnic Institute}
  \city{Troy}
  \state{NY}
  \country{USA}}
\email{yener@cs.rpi.edu}

\author{Lirong Xia}
\affiliation{
  \institution{Rensselaer Polytechnic Institute}
  \city{Troy}
  \state{NY}
  \country{USA}}
\email{xial@cs.rpi.edu}

\begin{abstract}
We develop a game theoretic model of malware protection using the state-of-the-art sandbox method,  to characterize and compute optimal defense strategies for anti-malware. We model the strategic interaction between developers of malware (M) and anti-malware (AM) as a two player game, where AM commits to a strategy of generating sandbox environments, and M responds by choosing to either attack or hide malicious activity based on the environment it senses. We characterize the condition for AM to protect all its machines, and identify conditions under which an optimal AM strategy can be computed efficiently. For other cases, we provide a quadratically constrained quadratic program (QCQP)-based optimization framework to compute the optimal AM strategy. In addition, we identify a natural and easy to compute strategy for AM, which as we show empirically, achieves AM utility that is close to the optimal AM utility, in equilibrium.
\end{abstract}


\newcommand{\BibTeX}{\rm B\kern-.05em{\sc i\kern-.025em b}\kern-.08em\TeX}

\begin{document}


\pagestyle{fancy}
\fancyhead{}


\maketitle 

\section{Introduction}
Malicious programs, or malware (M) for short, are a threat to individuals, companies, and even military units and intelligence agencies. One approach to stopping malware is to scan suspected programs with an anti-malware (AM) program, which can check the suspected program against a list of known, bad programs. In response, malware resort to manipulating its own code dynamically during runtime to bypass static signature scanning. To counter this, AM programs execute suspected programs in a contained, simulated environment, called a {\em sandbox}, in an attempt to trick the malware into showing its true malicious activity. This allows the AM to check program behavior at run-time against a list of known bad behaviors. In response, malware now checks its environment to make an intelligent decision on whether to {\em attack} by unleashing its true malicious activity~\cite{Bulazel2017:Survey}. This raises the question, {\em what is the optimal way for AM to protect machines defended by it, from malware developers who will study its behavior?}

The tools of {\em game theory} are well suited to analyze the strategic interaction between malware developers (M) and anti-malware developers (AM) in our setting of {sandboxing}, where AM's strategies involve the use of sandbox environments to dynamically analyze malware. Despite the large literature and wide application of game theoretic analysis of problems in physical~\cite{Tambe2011:Security} and cyber~\cite{Roy2010:Survey,Singh2010:Malware,Do2017:Game} security, sandboxing has not been analyzed from a game-theoretic perspective.

We use computational game-theoretic methods to develop strategies and guidelines to improve the state of the art in sandbox analysis. Unfortunately, existing literature on game theoretic analysis of malware do not model the sandboxing problem, because the dynamics of the game and utility functions we study in this paper are quite different from previous work. See Section~\ref{sec:related} for more details. We address the following key research question:

\vspace{2mm} {\em How can we model sandboxing as a game and characterize and compute optimal strategies for the anti-malaware?}

\begin{figure}[htp]
	\centering
	\includegraphics[width=.9\linewidth]{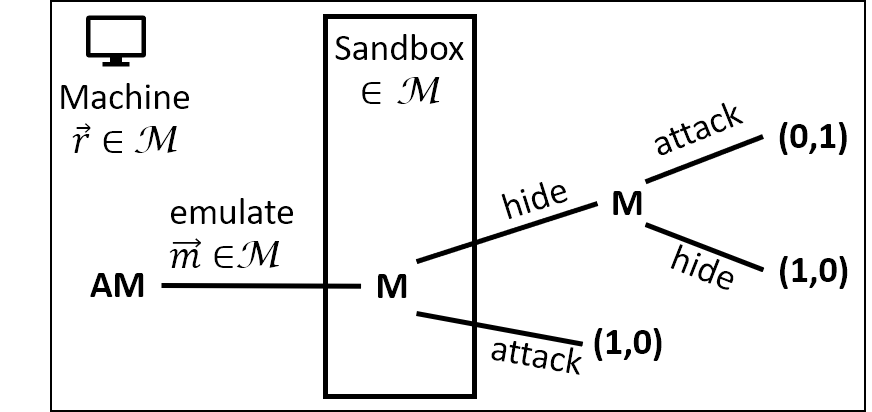}
	\caption{\label{fig:simple} High level depiction of the timing of the game.}
\end{figure}

\subsection{Our Contributions}
Our main conceptual contribution is the first game theoretic model of the strategic behavior between sandboxing AM and M. Figure~\ref{fig:simple} illustrates the timing and actions of AM and M on a real machine that AM defends. 

Suppose an AM is installed on some (not necessarily all) machines. The sandbox game has three stages. On each machine that AM defends, the AM generates a sandbox according to a strategy $\pi$ within which to analyze the potential malware. Given any machine environment $\vm$ that M is executed in, M decides an action of attacking or hiding according to her chosen strategy $\rho$. The three stages are:

\begin{enumerate}[label={\bf Stage \arabic*:},leftmargin=*]

\item AM commits to a sandboxing strategy, when it is deployed on the real machines it defends.

\item M observes AM's strategy and responds with an attack strategy $\rho$.

\item If M bypasses the sandbox in stage 2, then it will be run on the real machine, and attack according to $\rho$.

\end{enumerate}

We note that in stage 3, M uses the same attack strategy $\rho$ as it does in stage 2. This is because M's decision is based only on the environment it perceives during execution, and M cannot distinguish between whether an environment corresponds to a real machine or one that is being emulated by AM in a sandbox.

M's goal is to successfully attack as many real machines as possible. AM's goal is to {\em protect} the maximum number of machines it {\em defends} as possible. A real machine is protected if either M attacks in the sandbox and is caught, or M bypasses the sandbox but chooses not to attack the real machine.

Our main technical contributions are analytical solutions of AM-optimal strategies in equilibrium for a wide range of cases, and practical algorithms to compute such solutions otherwise. We summarize our results as a set of guidelines for sandboxing AMs in Section~\ref{sec:guidelines}. Our main results are:

\begin{itemize}[leftmargin=*]

\item When the AM defends every machine, we identify natural and easy to compute AM-optimal solutions (Theorems~\ref{thm:installednaive} and~\ref{thm:installedsophisticated}). These results are summarized in Table~\ref{tbl:contrib_everywhere}.

\item When AM defends at most half of all real machines, there is an equilibrium solution where AM protects all machines which it defends, as we show in Theorem~\ref{thm:partial}. A potential application of this finding is that the deployment setting may be modified by creating undefended dummy virtual machines that emulate real machines, which superficially increase the number of machines where AM is not installed.
	
\item We provide a quadratically constrained quadratic program (QCQP) framework in Algorithm~\ref{alg:qcqp} for computing an AM-optimal solution in equilibrium for real world settings where no analytical solution is known.

\item We show through theoretical results or experiments that natural strategies of sandboxing are effective. For example, even on cases where no analytical solution is known, generating sandbox environments using a distribution identical to the the distribution of real machines in existence yields good AM utility in practice.

\end{itemize}

\section{Related Work}\label{sec:related}

\noindent{\bf Anti-analysis.}
{\em Sandbox} is a general term referring to any simulated or contained environments, such as emulation, virtualization, or debuggers. Although these tools intend to mimic the original environment, they often do not do so perfectly. These imperfections allow programs that can detect them to distinguish between running within sandboxes from running within normal environments. This gives programs the capability of decision, specifically to act benignly while in a sandbox and act maliciously when running in the normal environment, in an attempt to distinguish the real environment from the artificial, sandbox environment. This results in a game of the malware inspecting the sandbox as the anti-malware attempts to inspect the program. Both parties are attempting inspection and classification.

Bulazel et al.~\cite{Bulazel2017:Survey} did a recent survey of automated malware analysis and evasion.
There is extensive literature cataloging the various ways to detect emulators, virtual machines, and debuggers \cite{Raffetseder2007,Branco2012}. There is also extensive work in engineering solutions that approximate transparent introspection of programs \cite{Deng:2013:SSB:2523649.2523675,balzarotti10:split_personality}.

\vspace{1em}
\noindent{\bf Formalization.}
Blackthorne et al.~\cite{Blackthorne2016} define a formal framework for environmentally sensitive malware in which they show the conditions necessary for malware to distinguish sandboxes and remain undetected. Blackthorne et al.~\cite{ENVKEY_LATINCRYPT2017} formalize malware interacting with and distinguishing environments through the lens of cryptography with the label environmental authentication. They analyze the interaction in terms of correctness, soundness, and key sources, but they do not suggest any strategies for game players to use possible equilibria. Dinaburg et al.~\cite{Dinaburg2008} present a formalization for transparent malware analysis. {\em Transparent} analysis means using an environment which is impossible to detect by the program being analyzed, i.e., the analysis environment is transparent. Kang et al.~\cite{Kang2009} also briefly formulate the problem of transparent malware analysis within emulators, but do not offer any further investigation than the beginnings of describing transparency.

\vspace{1em}
\noindent{\bf Game Theory.}
To the best of our knowledge, our work is the first to formally model the interaction between M and AM via sandboxing through the lens of game theory. Game theory has been used to model various problems in cybersecurity. (See~\cite{Singh2009:Optimal,Singh2010:Malware} and~\cite{Do2017:Game,Roy2010:Survey} for recent surveys). However, relatively little literature exists on game theoretic analysis of the interactions between the developers of malware and anti-malware.

\citeauthor{Sinha2015:Physical}~\cite{Sinha2015:Physical} discusses techniques developed to address the challenges of scale, uncertainty, and imperfect observation by the attacker in security games modeled as Stackelberg games, and applications to cybersecurity. Stackelberg games have also been used to model security games involving web applications~\cite{Vadlamudi2016:Moving}, network security~\cite{Vanvek2012:Game,Durkota2015:Game}, competition among multiple malwares~\cite{Lee2015:Host}, and audit games~\cite{Blocki2013:Audit,Blocki2015:Audit}.

Our model has a number of differences with previous work. First, the AM must commit to the same strategy on every machine of the same type. Unlike security games, when AM can commit to mixed strategies, there are an infinite number of pure strategies for even a naïve AM which uses the same strategy on every machine. Each pure strategy is represented by a vector with a component for each type of machine whose value correspond to the probability with which AM creates a sandbox of that type. This means that unfortunately, standard algorithms and techniques for solving security games do not apply to the sandbox game.

\begin{figure*}[h!]
	\centering
	\begin{gather*}
	u_{\m}(\pi,\rho) = \sum_{\vr \in \mm}\underbrace{\left( 1-\vd_{\vr}\right)\vrho_{\vr}}_\text{M attacks $|$ No AM} +\allowbreak \underbrace{\vd_{\vr}\Big[1 - \sum_{\vm \in \mm}\vpi_{\vm}^{\vr}\vrho_{\vm}\Big]\vrho_{\vr}}_\text{M evades sandbox and attacks real machine $|$ AM installed}\\
	u_{\am}(\pi,\rho) = \left.\sum_{\vr \in \mm}\underbrace{\vd_{\vr}}_\text{AM installed}\Bigg[\sum_{\vm \in \mm}\allowbreak\Big[\underbrace{\vpi^{\vr}_{\vm}\vrho_{\vm}}_\text{M caught in sandbox} +\allowbreak \underbrace{\vpi^{\vr}_{\vm}(1-\vrho_{\vm})(1-\vrho_{\vr})}_\text{M evades sandbox, then hides}\Big] + \underbrace{\Big(1 - \sum_{\vm \in \mm}\vpi^{\vr}_{\vm}\Big)(1-\vrho_{\vr})}_\text{No sandbox. M hides}\Bigg]\right.
	\end{gather*}
	\caption{\label{fig:utility}Utility functions of AM and M.}
\end{figure*}

\section{Modeling the Sandbox Game}\label{sec:prelim}
We are given a set $\mm$ of all types of environments, where each $\vm \in \mm$ is represented by a $k$-vector of features (i.e.,~emulated clock time, fingerprint, network connection, etc.) that fully and uniquely describes machines of type $\vm$. Every {\em real machine} presents some environment in $\mm$ natively. Similarly, every sandbox generated by AM presents also presents an environment in $\mm$. Table~\ref{tab: notation} summarizes the notation used throughout this paper.\\

\noindent{\bf Real World Settings} are described by $|\mm|$-vectors $\ve$ and $\vd$. The $\vr$-th component of $\ve$, denoted $\ve_{\vr}$ is the fraction of all machines in existence in the real world that are of type $\vr\in\mm$. The $\vr$-th component of $\vd$ (respectively $\vone - \vd$) is the fraction of all real machines that are of type $\vm$ and are {\em defended} (respectively {\em not defended}) by AM. In addition, we will use $D$ and $1-D$ to denote the proportion of all real machines that defended, and not defended by AM, respectively.\\

\noindent{\bf AM's Strategy Space} is the set $\Pi$ of all functions $\pi$ that map each real machine type $\vr\in \mm$ to an $|\mm|$-vector $\pi(\vr) = \vpi^{\vr}$ which describes a distribution over $\mm$. On any given real machine of type $\vr\in\mm$, the $\vm$-th component of $\vpi^{\vr}$, denoted $\vpi^{\vr}_{\vm}$, gives the probability with which AM generates a sandbox of type $\vm\in\mm$ on a real machine of type $\vr$ when AM uses strategy $\pi$. The probability that an AM using a strategy $\pi$ does not generate a sandbox on a real machine $\vr$ is $1 - \sum_{\vm\in \mm} \vpi^{\vr}_{\vm}$.\\ 

\noindent{\bf Malware's Strategy Space} is the set $\Rho$ of all vectors in $[0,1]^{|\mm|}$. Each strategy $\rho \in \Rho$ is represented by an $|\mm|$-vector $\vrho$. The $\vm$-th component of $\vrho$, denoted $\vrho_{\vm}$ is the probability with which M attacks when presented with an environment $\vm\in\mm$. Note that M cannot distinguish between a sandbox and a real machine that present the same environment. M's strategy only depends on the execution environment $\vm$ presented to it, regardless of whether it is an environment presented by a real machine or by a sandbox.\\

\noindent{\bf Restrictions on strategy space.}
We say that AM is {\em na\'ive} if it generates sandboxes with the same probability distribution on every type of real machine, i.e., when AM is restricted to a strategy $\pi \in \Pi$, where for every pair of real machine types $\vr, \vr' \in \mm$,  it holds that $\vpi^{\vr} = \vpi^{\vr'}$. We say that AM is {\em deterministic} if for every machine $\vr \in \mm$, there is a sandbox environment $\vm$ which is deterministically always generated by AM in $\vr$. Otherwise,  AM is {\em non-deterministic}. We say that AM is {\em sophisticated}, if AM can choose different distributions over $\mm$ to create sandboxes on any real machine, i.e., its strategy space is all of $\Pi$. 

Similarly, M is said to be naïve, if its probability of attack does not depend on the environment it perceives, i.e., for every pair of possible environment types $\vm, \vm' \in \mm$, it holds that $\vrho_{\vm} = \vrho_{\vm'}$. M is {\em deterministic} if M is restricted to strategies $\rho$, where for each $\vm \in \mm$, $\vrho_{\vm} \in \{0,1\}$, i.e., in each environment M can either only always attack or always hide. Otherwise, M is {\em non-deterministic}. M is sophisticated when its strategy space is all of $\Rho$.

We note that in our formulation,  for AM and M, deterministic strategies and non-deterministic strategies are pure strategies. In other words, there are infinite pure strategies,  each of which assigns a number between $0$ and $1$ to each machine type, and deterministic strategies simply mean that each component is either $0$ or $1$ but not in-between.\\

\noindent {\bf Utility Function.} AM wants to protect machines it defends (but not all machines). AM's payoff (utility) is the expected fraction of real machines which it defends that are not attacked. M's payoff is the expected fraction of machines that are successfully attacked, regardless of whether the machine is defended by AM. Consequently, given mixed strategies $\pi$ and $\rho$ of AM and M respectively, we define M's utility $u_{\m}(\pi,\rho)$ and AM's utility $u_{\am}(\pi,\rho)$ as shown in Figure~\ref{fig:utility}. Notice that if M is caught in a sandbox, it is not allowed to execute on the real machine, and therefore cannot attack the real machine. Notice also that payoffs are bound between $0$ and $1$.\\

\noindent {\bf Solution Concept.} We will focus on characterizing and computing the pure strategy subgame perfect equilibrium (SPNE) of the game. This will give us the optimal strategy for AM while M can perfectly respond to AM's strategies. Computing an AM optimal SPNE solution involves solving the following optimization problem:
\begin{align*}
\max_{\pi,\rho} ~& u_\am(\pi,\rho)\\
\text{such that} ~& \rho \in \arg\max_{\hat\rho\in \Rho}u_\m(\pi,\hat\rho)
\end{align*}
Unfortunately, this involves solving multiple non-convex quadratically constrained quadratic programs which are generally NP-hard. To deal with this, we identify several tractable cases, and suggest practical algorithms for computing the solution otherwise.\\

\begin{table*}[htp!]
\centering
\begin{tabular}{ll}
\hline
Notation           & Meaning \\ \hline
$\mm$              & The set of all possible types of environments. \\
$\vm\in\mm$   & An environment in $\mm$, which may be presented by either a real machine or a sandbox. \\
$\vr\in\mm$   & An environment type in $\mm$ presented by a real machine. \\
$\ve$              & $|\mm|$-vector where $\ve_{\vm}$ denote the fraction of real machines of type $\vm\in\mm$. \\
$\vd$              & $|\mm|$-vector where $\vd_{\vm}$ denotes the fraction of real machines of type $\vm\in\mm$ that are defended by AM. \\
$D$                & Fraction of all real machines in existence that are defended AM. \\
$\vpi$             & AM's strategy: maps each possible real machine environment $\vr$ to a distribution over $\mm$ of generating sandboxes. \\
$\vpi^{\vr}_{\vm}$ & The probability that AM generates a sandbox of type $\vm$ given a real machine of type $\vr$. \\
$\vrho$            & M's strategy: An $|\mm|$-vector where $\vrho_{\vm}$ denotes the probability that M attacks when presented an environment $\vm$. \\
$u_\am, u_\m$      & Utility functions of AM and M. \\ \hline
\end{tabular}
\caption{A summary of notation used frequently throughout this paper.\label{tab: notation}}
\end{table*}

\noindent{\bf Natural Strategies.} We will evaluate the AM-optimal strategies against some natural strategies as described in Table~\ref{tbl:natural}. For example, in the \existence strategy, AM generates a sandbox of type $\vm$ with probability $\vpi_{\vm} = \veta_{\vm}$, i.e., the fraction of type $\vm$ machines. 
\begin{table}[h]
	\centering
	\small
	\begin{tabular}[width=\linewidth]{|c|c|c|}\hline
		Name & $\vpi_{\vm}$ & Distribution\\ \hline
		\existence  & $\veta_{\vm}$ & Real machines\\
		\defended & $\frac{\vetap_{\vm}}{D}$ & Defended machines\\
		\undefended & $\frac{1 - \vetap_{\vm}}{1-D}$ & Undefended machines\\
		\perdefended & $\frac{\vetap_{\vm}}{\veta_{\vm}}$ & $\propto$ \% Defended\\
		\perundefended & $\frac{1-\vetap_{\vm}}{\veta_{\vm}}$ & $\propto$ \% Undefended\\
		\majority & $\vpi_{\arg\max\veta_{\vm}} = 1$ & Majority\\
		\uniform & $\frac{1}{|\mm|}$ & Uniform\\ \hline
	\end{tabular}
	\caption{\label{tbl:natural}Natural strategies for AM.}
\end{table}

\section{SPNE When AM Defends Every Machine}\label{sec:everywhere}
\begin{table*}[htp]
	\centering
	\small
	\begin{tabular}{|c|c|c|c|c|c|} \hline
		\multicolumn{2}{|c|}{Strategy space} & \multirow{ 2}{*}{Players} & \multicolumn{2}{c|}{Equilibrium Strategy} & \multirow{ 2}{*}{Utility} \\ \cline{1-2}\cline{4-5}
		AM & M& & Machine $A$ ($\vec e_{\vec a} = 0.4$) & Machine $B$  ($\vec e_{\vec b} = 0.6$) &  \\ \hline
		\multirow{ 2}{*}{Naïve deterministic}& \multirow{ 2}{*}{deterministic} & AM  & 0 & 1 & 0.6 \\
		&&M  & 1 & 0 & 0.4 \\ \hline
		\multirow{ 2}{*}{Naïve}& \multirow{ 2}{*}{deterministic} &  AM  & 0.4 & 0.6 & 0.76 \\
		&&M  & 1 & 0 & 0.24 \\ \hline
		\multirow{ 2}{*}{Naïve}& \multirow{ 2}{*}{Sophisticated} &  AM & 0.4 & 0.6 & 0.75 \\
		&&M  & 0.5 & 0.5 & 0.25 \\ \hline
		\multirow{ 2}{*}{Sophisticated~deterministic}& \multirow{ 2}{*}{Sophisticated} &  AM  & Emulate A & Emulate B & 0.75 \\
		&& M  & 0.5 & 0.5 & 0.25 \\ \hline
	\end{tabular}
	\caption{\label{tbl:simple} Example of the equilibrium strategies and payoffs for AM and M for two machines $A$ and $B$ with $\vec e_{\vec a}= 0.4$ and $\vec e_{\vec b} = 0.6$.}
\end{table*}

We first consider the case where AM defends all machines. In this case $\vd_{\vm} = \ve_{\vm}$ for all $\m\in \mm$. We will prove that \existence is the optimal strategy to generate sandboxes when AM defends all real machines when AM is naïve in Theorem~\ref{thm:installednaive}. This involves randomly generating sandboxes according to the distribution $\ve$, i.e., for each real machine $\vr \in R$, and possible environment $\vm \in \mm$, $\vpi^{\vr}_{\vm} = \ve_{\vm}$. When AM can commit to a sophisticated strategy, the solution is even more simple and intuitive: {\em deterministically create a sandbox of the same type as the real machine being defended.} In equilibrium, AM is guaranteed a utility of $0.75$, as we prove in Theorem~\ref{thm:installedsophisticated}.

Note that in the special case where M is naïve, AM's optimal strategy is any strategy that always creates a sandbox on any real machine. This is because a naïve M cannot distinguish between different types of machines and the game degenerates into a single-type model. The example below shows different cases where AM and a non-naïve M are restricted on different strategy spaces. 
\begin{ex}\rm
	\label{ex:simple}
	The simple example of Table~\ref{tbl:simple} illustrates the following natural message about the impact of modeling choices of the game on the payoff of AM. {\bf In general, allowing AM more flexibility, or placing more restrictions on M leads to a higher payoff for AM in equilibrium.} At a high level, these roughly correspond to giving AM more resources and constraining the resources of M.
	
	Table~\ref{tbl:simple} illustrates an example of the various equilibrium strategies and payoffs for a setting with two machines $A$ and $B$ with $\vec e_{\vec a}= 0.4$ and $\vec e_{\vec b} = 0.6$. Specifically, the {\em Equilibrium Strategy} column shows the strategies, where its first column is $\vpi_{\vec a}$ for AM and $\vrho_{\vec a}$ for M; and its second column is $\vpi_{\vec b}$ for AM and $\vrho_{\vec b}$ for M. The exception is the fourth row where AM takes a sophisticated deterministic strategy. Here AM's strategy is to create a sandbox that emulates the real machine it is defending, by creating a sandbox of type $A$ for all real machines $A$ and a sandbox of type $B$ for all real machines $B$.
	
	When AM uses non-deterministic sandboxing strategies to fight against a deterministic-strategy-only M, AM's utility increases from $0.6$ to $0.76$ (first row vs.~second row in Table~\ref{tbl:simple}). If M is also allowed to be non-deterministic, then AM's payoff in equilibrium reduces slightly to $0.75$ (third row in Table~\ref{tbl:simple}). As another example, allowing AM to use sophisticated (but still deterministic) strategies to fight against M increases the utility from $0.6$ to $0.75$ (first row vs.~fourth row in Table~\ref{tbl:simple}). 
	
	We consider the setting with two machines $A$ and $B$ represented by feature vectors $\vec a$ and $\vec b$ respectively. We start by restricting M's ability to be deterministic only. 
	
	{\em Naïve Deterministic vs.~Naïve Nondeterministic AM} (first row vs. second row in Table~\ref{tbl:simple}).
	Allowing the flexibility of randomized sandbox generation increases AM's utility to 0.76 from 0.6 when AM is restricted to deterministic generation. Suppose AM assigns probability $\pi\in [0,1]$ to emulate $\vec a$ and $1-\pi$ to emulate $\vec b$. M's payoff becomes: $$\max\underbrace{(0.4\times(1-\pi)}_{\text{Always attack }\vec a}, \underbrace{0.6\times \pi)}_{\text{Always attack }\vec b}$$ 
	
	Note that when $M$'s strategy is to always attack $u_M=0$ since AM always creates a sandbox, and that when $M$ always does not attack $u_M=0$ again. It is easy to see that $u_M$ is minimized when $\pi=0.4$, i.e., AM mimics the distribution of real machines.
	M is indifferent between always attacking $\vec a$ or always attacking $\vec b$, which results in a payoff of $0.76$ to AM.
	
	{\bf Remark}: When AM is Naïve and deterministic and M is sophisticated, the best strategy for AM is \majority. Note that in this case AM's strategy is to pick a single type $\vm$ and create a sandbox of type $\vm$ on all real machines. Then, M's best response is attacking $\vm$ with probability 0.5, and always attacking other types. Therefore, $u_M = 1-0.75\ve_{\vm}$ and $u_{AM} = 0.75\ve_{\vm}$. Therefore, AM's utility is maximized when it picks the type with the highest proportion to emulate as a sandbox environment, which is exactly the strategy \majority.
	
	{\em Sophisticated and non-deterministic M} 
	(third row vs. second row in Table~\ref{tbl:simple}). Suppose AM assigns probability  $\pi \in [0,1]$ to emulate $\vec a$ and $1-\pi$ to emulate $\vec b$, and suppose M chooses to attack with probability $\rho_A$ (respectively, $\rho_B$) in environment $\vec a$  (respectively, $\vec b$). M's payoff becomes 
	\begin{align}\label{eq:basicpayoff}
	\begin{split}
	\underbrace{0.4\times (\pi(1-\rho_A)+(1-x)(1-\rho_B))\times \rho_A}_\text{M succeeds on machine A}+\\ \allowbreak\underbrace{0.6\times (\pi(1-\rho_A)+(1-\pi)(1-\rho_B))\times \rho_B}_\text{M succeeds on machine B}.
	\end{split}\end{align}
	
	While this formula seems hard to solve analytically, observe that no matter what $\pi$ is, M's payoff is always $0.25$ when it chooses $\rho_A=\rho_B=0.5$. Also notice that when $\pi=0.4$, Equation~(\ref{eq:basicpayoff}) becomes $(0.4 \rho_A+0.6\rho_B)(0.4 (1-\rho_A)+0.6(1-\rho_B))$, which is maximized at $\rho_A=\rho_B=0.5$ giving the equilibrium shown in Table~\ref{tbl:simple}.

	\vspace{2mm}\noindent{\em Sophisticated deterministic AM vs.~Nondeterministic M.} This is the last row in Table~\ref{tbl:simple}. Similarly to the na\'ive mixed vs.~mixed case, M can choose $\rho_A=\rho_B=0.5$ to guarantee a payoff of $0.25$. AM's optimal strategy now becomes always  emulating the machine being defended. \hfill$\blacksquare$
\end{ex}

Theorems~\ref{thm:installednaive} and~\ref{thm:installedsophisticated} follow from the observation that (1) when AM defends all machines, we are faced with a zero-sum game, and (2) identifying an equilibrium under which AM achieves the highest possible utility under any equilibrium. 

\begin{thm}\label{thm:installednaive} When AM is installed on every real machine, and AM is naïve, and can commit to mixed strategies, the natural strategy \existence, is optimal in equilibrium.
\end{thm}

\begin{proof}
When AM is naïve, it takes the same distribution for all types of machine $\vr \in \mm$. Therefore, we use $\vpi$ to denote AM's strategy. Then we can rewrite the utility of M  $u_M(\pi, \rho)$ as the following formula: $u_M(\pi,\rho) = (\sum_{\vr\in \mm} \ve_{\vr} \vrho_{\vr}) ( 1-\sum_{\vm\in \mm} \vpi_{\vm} \vrho_{\vm})$. Note that since all machines are defended, we have $\vd_{\vr} = \ve_{\vr}$ for all $\vr\in\mm$. We also have $u_{AM}(\pi, \rho) = 1-u_M(\pi, \rho)$.

We begin by noting that for {\em any} AM strategy $\pi$, $M$ can guarantee a utility of at least $0.25$ when $M$ uses a naïve strategy of attacking with probability $0.5$, because $\sum_{\vr\in\mm} \ve_{\vr} =1$,  $\sum_{\vm\in\mm} \vpi_{\vm} \le1$ and $M$'s utility is
$$
u_M(\pi,\rho) = (0.5\sum_{\vr\in \mm} \ve_{\vr}) ( 1-0.5\sum_{\vm\in \mm} \vpi_{\vm}) \ge 0.5\times(1-0.5) =0.25.
$$
Therefore, for any choice of $\pi$ for AM, if M selects a best response, $u_{AM}(\pi, \rho) \le 0.75$.

Now, consider the case where AM adopts the strategy \existence $\pi^*$ where $\vpi^*_{\vm} = \ve_{\vm}$ for every $\vr \in \mm$. Then $u_M(\pi,\rho) = \big(\sum_{\vec r \in \mm}\ve_{\vr}\vrho_{\vr}\big)\big(1-\sum_{\vm \in \mm}\ve_{\vm}\vrho_{\vm}\big)$.
In order to maximize $u_M$, we consider its derivative:
$$
\frac{\partial u_M}{\partial \vrho_{\vr}} =  \ve_{\vr} \big(1-2\sum_{\vm \in \mm}\ve_{\vm}\vrho_{\vm}\big)
$$
Note that when $\vrho_{\vm} =0.5$ for all $\vm\in\mm$ (M attacks every type with probability 0.5), $\frac{\partial u_M}{\partial \vrho_{\vr}} =0$ for every $\vm\in\mm$. and $u_M(\vpi, \vrho)$ is maximized. Then we have $u_M=0.25$ and $u_{AM}=0.75$ which is the highest utility that AM can obtain in equilibrium. Therefore, \existence is an AM-optimal strategy in equilibrium. 
\end{proof}

\begin{thm}\label{thm:installedsophisticated} When AM is installed on every real machine, and
AM is sophisticated, an AM optimal strategy in equilibrium is to emulate the current machine, i.e., $\pi(\vec r,\vec r) = 1$ for every real machine $\vec r$.
\end{thm}

\begin{proof}
    
When AM is sophisticated, 
$u_M(\pi,\rho) = \allowbreak \sum_{\vec r \in \mm}\ve_{\vr}\vrho_{\vr} \big(1-\sum_{\vm \in \mm}\vpi^{\vr}_{\vm}\vrho_{\vm}\big)$. First, we note again that if M uses a naïve strategy of attacking with probability $0.5$, we have 
$$
u_M(\pi,\rho) =0.5\sum_{\vec r \in \mm}\ve_{\vr} \big(1-0.5\sum_{\vm \in \mm}\vpi^{\vr}_{\vm}\big) \ge 0.25\sum_{\vec r \in \mm}\ve_{\vr} = 0.25.
$$
Therefore, M's utility is at least $0.25$ for any AM strategy, and $u_{AM}(\pi, \rho)\le 0.75$ for any $\pi$ whenever M takes a best response. 

Then we only need to show that M attacking with probability 0.5 is the best response when AM  set $\pi_{\vr}^{\vr} = 1$ for every real machine $\vec r$. Note that in this case, M's utility is
 $u_M(\pi,\rho) = \sum_{\vec r \in \mm}\ve_{\vr}\vrho_{\vr} \big(1-\vrho_{\vr}\big)$. We can easily find that this is maximized at $\vrho_{\vr}=0.5$ for every $\vr\in\mm$, and $u_M=0.25$. 
Therefore, $u_{AM} = 0.75$  when AM set $\vpi_{\vr}^{\vr} = 1$, which is the maximum AM can achieve. Therefore, we conclude that  $\vpi_{\vr}^{\vr} = 1$ is an optimal strategy for AM in equilibrium.
\end{proof}

\section{Computing SPNE When There Are Undefended Machines}
\label{sec:partial}
In this section, we consider the case where AM defends at most half the real machines, there is an AM-optimal strategy where every machine that AM defends is protected in equilibrium, and \undefended is one such strategy.

\begin{thm}\label{thm:partial}
	When at most half of all machines are defended ($D\le \frac12$), there exists an equilibrium solution where AM protects every defended machine (AM utility is $D$), and M always attacks. Specifically, \undefended is an AM-optimal strategy in equilibrium, for which M's best response is to always attack.
\end{thm}
\begin{proof}
	We start by deriving the utility M receives if it always attacks. We start by noting that given AM's naïve strategy $\pi$, M evades $D(1-\sum_{\vm \in \mm}\vpi_{\vm}\vrho_{\vm})$ sandboxes in expectation. However, M does not derive utility from machines on which it bypasses the sandbox but does not attack on the real machine which occurs with expectation $(1-\sum_{\vm \in \mm}\vpi_{\vm}\vrho_{\vm})\sum_{\vr \in \mm}\vd_{\vr}(1-\vrho_{\vr})$. Thus, we rewrite
	M's utility as: 
	 $u_\m = \sum_{\vm \in \mm}\left(1-\vd_{\vm}\right)\vrho_{\vm} +\allowbreak D(1-\sum_{\vm \in \mm}\vpi_{\vm}\vrho_{\vm}) \allowbreak - (1-\sum_{\vm \in \mm}\vpi_{\vm}\vrho_{\vm})\sum_{\vr \in \mm}\vd_{\vr}(1-\vrho_{\vr}) = \sum_{\vm \in \mm}\vrho_{\vm}(1 - \vd_{\vm} - D\vpi_{\vm}) + D - (1-\sum_{\vm \in \mm}\vpi_{\vm}\vrho_{\vm})\allowbreak\sum_{\vr \in \mm}\vd_{\vr}(1-\vrho_{\vr}).   $

	Note that whenever $D\vpi_{\vm} \le \left(1-\vd_{\vm}\right)$ for every $\vm \in \mm$, we have that $u_M$ is maximized when $\vrho_{\vm} = 1$ for every $\vm \in \mm$. This is because $\vrho_{\vm} \in [0,1]$ for every $\vm \in \mm$, and $D\vpi_{\vm} - (1-\vd_{\vm}) \le 0$ for every $\vm \in \mm$. Notice that when M always attacks, AM gets utility $D$ when it always creates a sandbox, which is the maximum possible utility AM can get, and therefore optimal.
	
	It is easy to see that when AM uses the natural strategy \undefended, i.e., setting $\vpi_{\vm} =\frac{1-\vd_{\vm}}{1-D}$, the condition $D\vpi_{\vm}\le \left(1-\vd_{\vm}\right)$ is satisfied for every $\vm$, and that AM creates a sandbox on every defended real machine, giving AM a utility of $D$ in equilibrium. Therefore, \undefended is an optimal strategy.
\end{proof}

\begin{ex}\rm\label{ex:undefended}
	We now consider a world with three types of machines A, B, and C, with features $\vec a,\vec b,$ and $\vec c$, respectively. $10\%$ of real machines are of type $\vec a$, $20\%$ of type $\vec b$, and the rest are of type $\vec c$. AM is installed on $70\%$ of machines of type $\vec a$ and type $\vec b$, and $30\%$ machines of type $\vec c$. 
	
	When AM's strategy is \undefended, M's best response is to always attack, allowing AM to protect 100\% of machines it defends. In comparison, AM's utility decreases to protecting only $92.2\%$ of machines it defends, if it uses the \existence strategy. This is in sharp contrast to our findings for AM optimal solutions when AM defends all machines that we observed in Section~\ref{sec:everywhere}.
	
	Here, we give a detailed analysis of M's utility when AM uses the strategy \existence. We list the relevant values in Table~\ref{tbl:ex2}. 
	
	\begin{table}[htbp]
\begin{tabular}{|l|l|l|l|l|}
\hline
Type $\vr$ & $\ve_{\vr} = \vpi_{\vr}$ & $\vd_{\vr}$ & $1 -\vd_{\vr}$ & $D\vpi_{\vr}$ \\ \hline
A          & 0.1                     & 0.07        & 0.03        & 0.42          \\ \hline
B          & 0.2                     & 0.14        & 0.06        & 0.84          \\ \hline
C          & 0.7                     & 0.49        & 0.49        & 0.294         \\ \hline
Total      & 1                       & $D = 0.42$  & $1-D=0.58$    &               \\ \hline
\end{tabular}
    \caption{\label{tbl:ex2} Values for computing $u_M$ in Example~\ref{ex:undefended} when AM's strategy is \existence.}
    \end{table}
    
    When AM uses \existence, $\vpi_{\vr} = \ve_{\vr}$ for every $\vr\in\mm$. Note that for A and B, we have $1-\vd_{\vec a}<D\vpi_{\vec a}$ and $1-\vd_{\vec b}<D\vpi_{\vec b}$. On the other hand,  we have $1-\vd_{\vec c}>D\vpi_{\vec c}$.  As we show below, in the best response strategy $\vrho$ for M, $\vrho_{\vec a}$ and $\vrho_{\vec b}$ do not equal to 1, while $\vrho_{\vec c} =1 $. 
    
    We first compute the derivative of $u_M$ based on the formula used in the proof of Theorem~\ref{thm:partial}.
    \begin{align*}
        \frac{\partial u_M}{\partial \vrho_{\vm}} = & \left(1-\vd_{\vm} - D\vpi_{\vm}\right) + \vd_{\vm}(1-\sum_{\vr\in\mm}\vpi_{\vr}\vrho_{\vr}) + \vpi_{\vm}\sum_{\vr\in\mm}\vd_{\vr}(1-\vrho_{\vr})\\
        = & \ve_{\vm} - \vd_{\vm}\sum_{\vr\in\mm}\vpi_{\vr}\vrho_{\vr} - \vpi_{\vm}\sum_{\vr\in\mm} \vd_{\vr}\vrho_{\vr}. 
    \end{align*}
    Substituting the values in Table~\ref{tbl:ex2}, we get:
    \begin{align*}
        \frac{\partial u_M}{\partial \vrho_{\vec a}} = & 0.1 - 0.014\vrho_{\vec a} - 0.028\vrho_{\vec b} - 0.07\vrho_{\vec c}\\
        \frac{\partial u_M}{\partial \vrho_{\vec b}} = & 0.2 - 0.028\vrho_{\vec a} - 0.056\vrho_{\vec b} - 0.14\vrho_{\vec c}\\
        \frac{\partial u_M}{\partial \vrho_{\vec c}} = & 0.7 - 0.07\vrho_{\vec a} - 0.14\vrho_{\vec b} - 0.294\vrho_{\vec c}
    \end{align*}
    Note that $\frac{\partial u_M}{\partial \vrho_{\vec c}}$ is always positive. Therefore $\vrho_{\vec c} =1$, and M always attacks on $\vec c$. On the other hand, $\frac{\partial u_M}{\partial \vrho_{\vec a}}$ and $\frac{\partial u_M}{\partial \vrho_{\vec b}}$ equals to zero simultaneously when $ 0.014\vrho_{\vec a} + 0.028\vrho_{\vec b} =0.03$. (Note that $\frac{\partial u_M}{\partial \vrho_{\vec b}}= 2\frac{\partial u_M}{\partial \vrho_{\vec a}}$). Therefore, $\vrho_{\vec a} = \vrho_{\vec b} =\frac57$, and $(\frac57, \frac57, 1)$ is an optimal strategy for M.
    
    Next, we consider the fraction of defended machines AM protects. Note that due to M's optimal strategy, M can only bypass the sandbox when AM generates either a type A or B sandbox. This means M can bypass the sandbox with probability $\frac27$. Then, the probability that a defended machine is attacked by M is:
    $$
    0.3\times\frac27\times (0.3\times\frac57 + 0.7\times 1) \approx 0.078
    $$
    Therefore, when AM adopts the \existence strategy, it successfully protects only 92.2\% of machines it defends. 
    \hfill$\blacksquare$
\end{ex}

\noindent{\bf Impact.} Theorem~\ref{thm:partial} raises interesting possibilities. One such possibility is to cheaply create sufficiently many virtual or dummy machines to manipulate the setting so that $D \le \frac12$, allowing us to trade off some computational resources for complete protection. At a high level, the idea is similar to honeypots in other cybersecurity problems.

\section{Computing M's Best Response}\label{sec:mbest}

\begin{table*}[htp!]
	\centering
	\small
	\begin{tabular}{|c|c|c|c|c|}\hline
		\multicolumn{2}{|c|}{AM}  & \multicolumn{2}{c|}{M} & \multirow{2}{*}{AM-optimal strategy} \\ \cline{1-4}
		Randomized & Sophisticated & Randomized & Sophisticated &  \\ \hline
		* & * & * & No & Always creating a sandbox \\ \hline
		No & No & * & Yes & \majority \\ \hline
		* & Yes & * & Yes & Emulate machine being defended  \\ \hline
		Yes & No & * & Yes & \existence \\ \hline
	\end{tabular}
	\caption{\label{tbl:contrib_everywhere} Guidelines when AM is installed on every real machine. $*$ indicates that added flexibility has no affect.}
\end{table*}

In this paragraph we provide a method to compute M's best response given a fixed AM strategy. This will be the basis of our QCQP framework. Before we proceed further, we rewrite the utility functions of AM and M in vector notation for convenience:
\begin{align*}
u_\m(\vpi,\vrho) = & \underbrace{(\vone - \vd)\cdot \vrho}_\text{M attacks | No AM} + \underbrace{\vd\cdot \vrho - \vd\cdot ((\Pi\cdot \vrho)\cvm \vrho)}_\text{M not caught, attacks | AM installed} \nonumber\\
= &~ \ve\cdot \vrho -  \vd\cdot ((\Pi\cdot \vrho)\cvm \vrho)\nonumber\\
u_\am(\vpi,\vrho) = &~\vd\cdot (1-\vrho) +  \vd\cdot ((\Pi\cdot \vrho)\cvm \vrho)
\end{align*}
where $\cvm$ is the component-wise multiplication operator between two vectors, and $\Pi$ is the matrix of AM's strategies, whose $\vr$-th row-vector is $\vpi^{\vr}$. 
For a fixed AM strategy $\pi$, we can solve for M's best response $\rho^*$ using the Lagrange multipliers for the optimization problem $\max_{\vrho}u_M(\vpi,\vrho)$ which are: 
$\forall \vm \in \mm, L_{\vm} = \frac{\partial u_\m(\vpi,\vrho)}{\partial \vrho_{\vm}} = \veta_{\vm} - 2\vetap_{\vm}\vpi_{\vm}\vrho_{\vm} -\allowbreak \sum_{\vec l \in \mathcal M \setminus \{\vm\}}(\vetap_{\vm}\vpi_{\vec l} + \vetap_{\vec l}\vpi_{\vm}) \vrho_l$.

Computing $\vrho^*$ involves computing the solutions to the following $3^{|\mm|} - 1$ systems of linear equations, and picking the solution that is feasible and maximizes M's utility. Each system of linear equations is indexed by an $|\mm|$-vector $\vc \in \{0,1,b\}^{|\mm|}$ (except $\vc = \vec 0$), where $b$ means ``between $0$ and $1$'', and is constructed by adding the $|\mm|$ equations as follows: \begin{enumerate*}[label=(\roman*)]\item $\vrho_{\vm}=0$ if $\vc_{\vm} = 0$, \item $\vrho_{\vm}=1$, if $\vc_{\vm} = 1$, or \item $L_{\vm}=0$ if $\vc_{\vm}=[0,1]$.\end{enumerate*}

\section{QCQP Framework for Computing AM-optimal SPNE}\label{sec:qcqp}
We now shift our focus to settings where $D>\frac12$.
Algorithm~\ref{alg:qcqp} provides a general quadratically constrained quadratic program (QCQP) formulation of the problem of computing AM-optimal SPNEs. We simultaneously solve for AM's strategy $\pi^*$ and M's strategy $\rho^*$ by setting $u_\am$ to be the objective, under the constraint that $\rho^*$ is the best response to the output strategy for AM $\pi^*$.

\begin{algorithm}[!ht]
	\begin{algorithmic}[1]
		\State {\bf Input:} A real world setting $\ve,\vd$.
		\State {\bf Output:} SPNE $\pi^*,\rho^*$.
		\State An empty set of SPNE solutions $S$.
		\For{each $\vc \in \{0,1,b\}^{|\mm|}$}
		\State Create a QCQP problem $P$ with variables $\vrho$, $\vpi$.
		\State Set the objective as $\max_{\vpi,\vrho} u_\am(\vpi,\vrho)$.
		\For{each $\vm$, $\vc_{\vm}$ add the binding constraints on $\vrho_{\vm}$}
		\State if $\vc_{\vm}=0$, add the constraint $\vrho_{\vm} = 0$.
		\State if $\vc_{\vm}=1$, add the constraint $\vrho_{\vm} = 1$.
		\State if $\vc_{\vm}=b$, add $\vrho_{\vm} \in [0,1]$, and $L_{\vm} = 0$.
		\EndFor
		\State Feasibility constraints $\vpi_{\vm}\allowbreak \in [0,1],\allowbreak \forall \vm$, and $\vpi.\vec 1 \le 1$.
		\State Compute the solution $\vpi, \vrho$.
		\State Test feasibility and constraint violations.
		\State Fix $\vpi$, and compute M's best response $\vrho'$ to $\vpi$.
		\If{$u_\m(\vpi,\vrho) \ge u_\m(\vpi,\vrho')$}
		\State Add $\vpi,\vrho$ as an SPNE to $S$.
		\EndIf
		\EndFor
		\Return Return the SPNE from $S$ with the highest AM utility.
	\end{algorithmic}
	\caption{\label{alg:qcqp} QCQP to compute AM-optimal SPNE.}
\end{algorithm}

Algorithm~\ref{alg:qcqp} involves (1) enumerating M's possible responses represented by the possible values that $\vc$ (line 7-10) can take on. For each $\vc$, the algorithm solves the QCQP problem to compute an optimal AM strategy $\vpi$ in equilibrium and the corresponding best response strategy $\vrho$ of M strategy (line 12), under the constraints $\vc$. Notice that here, $\vpi$ is an optimal AM strategy in equilibrium when M is constrained by $\vc$. In line 13, the algorithm tests constraint violations and discards the solution if some constraints are violated. This test is necessary because of limitations of current QCQP solvers. In lines 14-16, we fix the AM strategy to be $\vpi$, and compute $\vrho'$ which is M's best response to $\vpi$ without any constraints on M, using the technique described in Section~\ref{sec:mbest}, and test if $u_\m(\vpi,\vrho) \ge u_\m(\vpi,\vrho')$. If the inequality holds, the constrained best response $\vrho$ is also M's global best response, and $(\vpi, \vrho)$ is an equilibrium, and added to $S$. If the inequality does not hold, the algorithm discards the solution. Finally, the algorithm chooses the equilibrium with the highest AM utility from $S$. 

When AM is installed only on machines of a single type, we can solve for AM's optimal strategy efficiently as follows. The algorithm follows from the observation that $u_\am$ becomes a linear function in $\vpi$, with no critical points in the interval $[0,1]$. Thus we pick the solution which maximizes AM's utility from solving the $3^{|\mm|}-1$ sets of equations for every combination of setting $\vpi_{\vm}$'s to either $0$ or $1$ and the corresponding Lagrangian first order conditions on M's strategy.

\section{Guidelines for AM}\label{sec:guidelines}
Based on Theorems~\ref{thm:installednaive},~\ref{thm:installedsophisticated} and~\ref{thm:partial}, we refer to a setting as {\em easy}, if either AM defends all machines, or it defends at most half of all machines. In these cases, we have analytical solutions to AM optimal SPNE. We refer to other settings as {\em hard}, and solve them using our QCQP-based Algorithm~\ref{alg:qcqp}.

Here, we summarize our findings and provide guidelines for sandboxing AMs. Specifically, we answer the following questions through experiments: {\em Are there natural and easy to compute strategies for AM, assuming M best responds, without compromising utility when compared to the optimal strategy?} We consider some natural strategies which are summarized in Table~\ref{tbl:natural}.\\

\noindent{\bf Guidelines for Easy Settings.}
\begin{itemize}[leftmargin=*]
\item When AM defends every real machine, Table~\ref{tbl:contrib_everywhere} summarizes the AM-optimal strategies under various combinations of restrictions on strategy spaces of AM and M.
\item When AM defends less than half of all real machines, AM should use the \undefended strategy.
\item When AM defends a single type of real machine, AM should use the algorithm in Section~\ref{sec:qcqp} to compute an optimal strategy in equilibrium.
\end{itemize}

For hard settings, the QCQP SPNE computed using Algorithm~\ref{alg:qcqp} and the \existence strategy yield AM utility close to that of the BruteForce SPNE on average. \defended also performs consistently well in practice, yielding AM utility close to both QCQP and \existence strategies on average.

\subsection{Experimental Setup} To evaluate various sandboxing strategies for hard settings, we create a dataset of 1000 settings by generating settings involving machines of two types $A, B$ i.i.d. and retaining only the hard settings. For each hard setting, we compute AM's utility in the QCQP SPNE solution, a solution computed using brute force search which we call BruteForce, as well as AM's utility when AM plays each of the natural strategies in Table~\ref{tbl:natural} and M best responds.

We use the QCQP\footnote{https://github.com/cvxgrp/qcqp}~\cite{Park2017:General} extension of the CVXPY package using the suggest and improve method. For each subproblem, we compute $10$ solutions with random initial suggestions and pick the solution with the highest AM utility. Solutions are computed using the alternating direction method of multipliers (ADMM). We discard solutions with large constraint violations (Line 13 in Algorithm~\ref{alg:qcqp}). Then we fix AM's strategy and compute M's best response using the Lagrangian first order conditions as discussed in Section~\ref{sec:mbest}. We only retain solutions for which M's utility using the QCQP solution is within $0.01$ of M's utility using the best response (Line 14-16 in Algorithm~\ref{alg:qcqp}). 

\subsection{Experimental Results}

\begin{figure}[h]
	\centering
	\begin{tabular}{c}
		\includegraphics[width=.95\linewidth]{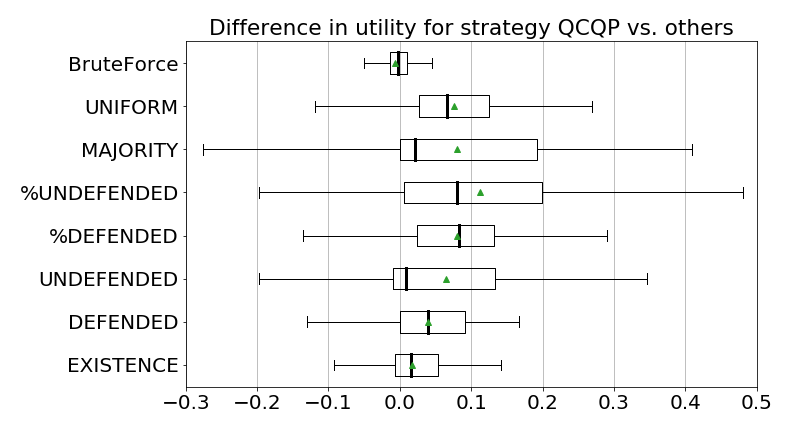}\\
		{\small (a) AM utility of QCQP solution against natural AM strategies.} \\
		\includegraphics[width=.95\linewidth]{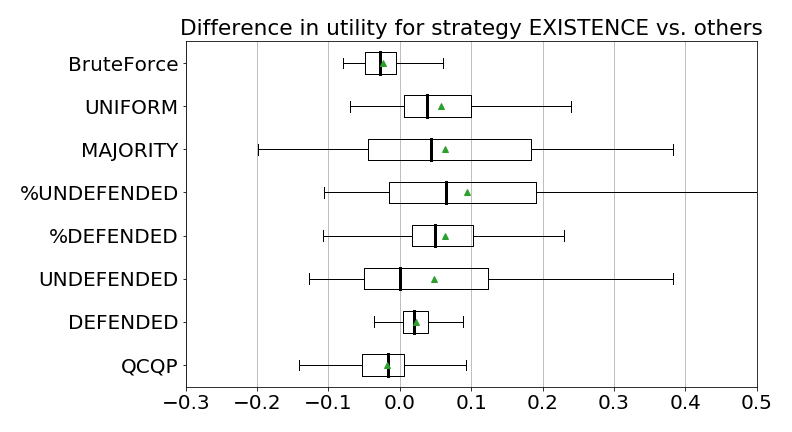}\\
		{\small (b) AM utility with real strategy against other AM strategies.}
	\end{tabular}
	\caption{\label{fig:comparison} Difference in AM utility using QCQP and \existence vs other natural strategies, and benchmarked against BruteForce. }
\end{figure}

In Figure~\ref{fig:comparison}, we report the differences between AM's utilities using the QCQP SPNE solution QCQP SPNE solution (Figure~\ref{fig:comparison} (a)), and when AM plays the \existence strategy and M best responds (Figure~\ref{fig:comparison} (b)), respectively against the AM utilities obtained when AM uses one of the other natural strategies, as well as AM utility in the SPNE computed using BruteForce as a benchmark. BruteForce was obtained by discretizing AM and M's strategy spaces in 0.01 intervals.\\

\noindent{\bf Guidelines for Hard Settings.}
\begin{itemize}[leftmargin=*]
\item \existence yields AM utility close to BruteForce strategy when M best responds.
\item On average AM utility with both \existence and QCQP strategies are close to AM utility in the BruteForce SPNE. Among all natural strategies, \existence consistently outperforms or matches other natural strategies, and even beats QCQP strategies on $31.023\%$ of simulated settings.
\item On average, both QCQP and \existence strategies are close to each  other in terms of AM utility with some exceptions where the solver failed to find a feasible solution that did not violate the first order conditions. The two methods complement each other, as shown in Figure~\ref{fig:comparison}, where we observe that they outperform each other depending on the setting.
\item If AM does not know about the distribution $\ve$ of real machines, sandboxing according to \defended is a viable alternative and works well on average.
\end{itemize}

\section{Conclusions and Future Work}
We provided the first game theoretic analysis of the sandbox game and the first theoretical results and practical algorithms for sandboxing. Specifically, our results provide concrete guidelines for deploying sandboxing AMs, allowing an AM developer to compute AM-optimal strategies under several natural restrictions on the strategy space of AM and M which correspond to different practical considerations in the deployment of AM and M. When AM either  defends every machine or defends fewer than half of all real machines, we identify natural and easy to compute strategies that are optimal for AM in equilibrium. The problem of computing an optimal AM strategy becomes harder when AM defends more than half of the real machines but not all of them. Our QCQP algorithm compute an optimal AM strategy but is computationally expensive. However, as we show empirically, the natural and easy to compute \existence strategy achieves AM utility that is close to optimal in practice.

There are several exciting avenues for future work: selecting AM strategies that are robust to M's selection of strategies in response to AM, and modeling the resource constraints as AM's ability to generate sandboxes, or M's inability to perfectly observe AM. Indeed, AM and M may often have different levels of information about a given deployment or constraints on the computational resources at their disposal. For example, commercially distributed AM may not be aware of the exact distribution of real machines where it may be deployed, and therefore be forced to commit to a naïve strategy. While our results already identify some natural and easy to compute AM-optimal solutions, fully exploring the impact of different constraints of information and computational resources on AM and M, and how to compute effective strategies for AM are an interesting question for future work.

\bibliographystyle{ACM-Reference-Format}  

\end{document}